\documentclass[]{article}

\usepackage[utf8]{inputenc}
\usepackage{amssymb}
\usepackage{amsthm}
\usepackage{float}
\usepackage{hyperref}
\usepackage{esvect}
\usepackage{fullpage}
\usepackage{stmaryrd}
\usepackage{subcaption}
\usepackage{mathtools}
\usepackage[dvipsnames]{xcolor}
\usepackage{tikz}
\usetikzlibrary{calc}
\usetikzlibrary{shapes.geometric}
\usetikzlibrary{shapes.misc}
\usetikzlibrary{positioning}
\usetikzlibrary{external}

\newcommand{\quot}[2]{{\raisebox{.2em}{$#1$}\left/\raisebox{-.2em}{$#2$}\right.}}
\newcommand{\interp}[1]{\left\llbracket#1\right\rrbracket}

\newcommand{\bvdots}{ \tikz[baseline, every node/.style={inner sep=0}]{ \node at (0,0){.}; \node at (0,-6pt){.}; \node at (0,6pt){.}; } }
\newcommand{\ket}[1]{|#1\rangle}
\newcommand{\bra}[1]{\langle#1|}
\newcommand{\braket}[2]{\langle#1|#2\rangle}
\newcommand{\df}{\stackrel{\scriptscriptstyle def}{=}}
\usetikzlibrary{circuits.ee.IEC}

\usepackage{tikzit}

\let\oldtikzfig\tikzfig
\renewcommand{\tikzfig}[1]{
	\tikzsetnextfilename{#1}
	\oldtikzfig{#1}
}

\let\oldctikzfig\ctikzfig
\renewcommand{\ctikzfig}[1]{
	\tikzsetnextfilename{#1}
	\oldctikzfig{#1}
}
\definecolor{zx_grey}{RGB}{211,211,211}

\tikzstyle{gn}=[fill=green, draw=black, shape=circle, tikzit category=ZX, tikzit fill=green, tikzit draw=black, tikzit shape=circle, inner sep=0.1em]
\tikzstyle{rn}=[fill=red, draw=black, shape=circle, tikzit fill=red, tikzit draw=black, tikzit category=ZX, tikzit shape=circle, inner sep=0.1em]
\tikzstyle{divide}=[regular polygon, regular polygon sides=3, shape border rotate=90, draw=black, fill={zx_grey}, inner sep=1.5pt, tikzit category=scal, rounded corners=0.8mm]
\tikzstyle{black}=[fill=black, draw=black, shape=circle, tikzit fill=black, tikzit draw=black, tikzit shape=circle, tikzit category=IH, inner sep=2pt]
\tikzstyle{gather}=[fill={zx_grey}, draw=black, tikzit category=scal, rounded corners=0.8mm, regular polygon, regular polygon sides=3, shape border rotate=-90, inner sep=1.5pt]
\tikzstyle{ggen}=[fill=white, draw=black, shape=rectangle, rounded corners=2mm, line width=1pt, tikzit draw=red, tikzit category=scal]
\tikzstyle{white}=[fill=white, draw=black, shape=circle, inner sep=2pt, tikzit category=IH]
\tikzstyle{mbox}=[fill=white, draw=black, rounded rectangle, rounded rectangle west arc=none, tikzit category=scal, tikzit shape=rectangle]
\tikzstyle{A}=[fill=white, shape=circle, tikzit category=scal, inner sep=1pt]
\tikzstyle{ggreen}=[fill=green, draw=black, shape=circle, tikzit category=SZX, tikzit fill=green, tikzit draw=black, line width=1pt, inner sep=0.1em]
\tikzstyle{gred}=[fill=red, draw=black, shape=circle, rounded corners=2mm, tikzit category=SZX, inner sep=0.1em, tikzit fill=red, line width=1pt]
\tikzstyle{ghad}=[minimum size=3mm, font={\scriptsize\boldmath}, shape=rectangle, inner sep=1mm, line width=1pt, outer sep=-1.5mm, scale=0.8, tikzit shape=rectangle, draw=black, fill=yellow, tikzit draw=blue]
\tikzstyle{boxm}=[fill=white, draw=black, rounded rectangle, tikzit category=scal, tikzit shape=rectangle, rounded rectangle east arc=none]
\tikzstyle{box}=[fill=white, draw=black, shape=rectangle]
\tikzstyle{had}=[fill=yellow, draw=black, shape=rectangle, tikzit category=ZX, tikzit fill=yellow, tikzit draw=black, inner sep=2.5pt]
\tikzstyle{gwhite}=[fill=white, draw=black, shape=circle, tikzit fill=white, tikzit shape=circle, line width=1 pt, inner sep=2 pt, tikzit draw=red]
\tikzstyle{gblack}=[fill=black, draw=black, shape=circle, tikzit fill=black, tikzit shape=circle, line width=1 pt, inner sep=2 pt, tikzit draw=red]
\tikzstyle{antipode}=[fill=red, draw=black, shape=rectangle, tikzit fill=red, tikzit draw=black, tikzit shape=rectangle, inner sep=2pt]
\tikzstyle{diamond}=[fill=white, draw=black, shape=diamond, inner sep=2pt]
\tikzstyle{mongr}=[fill=green, draw=green, shape=circle, inner sep=2pt]
\tikzstyle{monbl}=[fill=blue, draw=blue, shape=circle, inner sep=2pt]
\tikzstyle{bg}=[inner sep=0.7mm, minimum width=0pt, minimum height=0pt, fill=green, draw=white, very thick, shape=circle]
\tikzstyle{br}=[inner sep=0.7mm, minimum width=0pt, minimum height=0pt, fill=red, draw=white, very thick, shape=circle]
\tikzstyle{rmat}=[draw, signal, fill=red, signal to=east, signal from=west, inner sep=1pt, minimum height=6pt]
\tikzstyle{lmat}=[draw, signal, fill=red, signal to=west, signal from=east, inner sep=1pt, minimum height=6pt]
\tikzstyle{umat}=[draw, signal, fill=red, signal to=north, signal from=south, inner sep=1pt, minimum width=6pt]
\tikzstyle{dmat}=[draw, signal, fill=red, signal to=south, signal from=north, inner sep=1pt, minimum width=6pt]
\tikzstyle{box}=[shape=rectangle, text height=1.5ex, text depth=0.25ex, yshift=0.5mm, fill=white, draw=black, minimum height=5mm, yshift=-0.5mm, minimum width=5mm, font={\small}]
\tikzstyle{Z dot}=[inner sep=0mm, minimum size=2mm, shape=circle, draw=black, fill={rgb,255: red,160; green,255; blue,160}]
\tikzstyle{gdot}=[minimum size=3mm, font={\scriptsize\boldmath}, shape=rectangle, rounded corners=1.3mm, inner sep=1mm, outer sep=-1.8mm, scale=0.8, tikzit shape=circle, draw=black, fill=green, tikzit draw=blue]
\tikzstyle{X dot}=[Z dot, shape=circle, draw=black, fill={rgb,255: red,220; green,0; blue,0}]
\tikzstyle{rdot}=[minimum size=3mm, font={\scriptsize\boldmath}, shape=rectangle, rounded corners=1.3mm, inner sep=1mm, outer sep=-1.8mm, scale=0.8, tikzit shape=circle, draw=black, fill=red, tikzit draw=blue]
\tikzstyle{grdot}=[minimum size=3mm, font={\scriptsize\boldmath}, shape=rectangle, rounded corners=1.3mm, inner sep=1mm, line width=1pt, outer sep=-1.5mm, scale=0.8, tikzit shape=circle, draw=black, fill=red, tikzit draw=blue]
\tikzstyle{ggdot}=[minimum size=3mm, font={\scriptsize\boldmath}, shape=rectangle, line width=1pt, rounded corners=1.3mm, inner sep=1mm, outer sep=-1.5mm, scale=0.8, tikzit shape=circle, draw=black, fill=green, tikzit draw=blue]
\tikzstyle{arrow}=[-->]
\tikzstyle{rfarr}=[draw, signal, fill=black, signal to=east, signal from=west, inner sep=1pt, minimum height=6pt]
\tikzstyle{lfarr}=[draw, signal, fill=black, signal to=west, signal from=east, inner sep=1pt, minimum height=6pt]
\tikzstyle{ufarr}=[draw, signal, fill=black, signal to=north, signal from=south, inner sep=1pt, minimum width=6pt]
\tikzstyle{dfarr}=[draw, signal, fill=black, signal to=south, signal from=north, inner sep=1pt, minimum width=6pt]
\tikzstyle{ry}=[draw, signal, fill=yellow, signal to=east, signal from=west, inner sep=1pt, minimum height=6pt]
\tikzstyle{ly}=[draw, signal, fill=yellow, signal to=west, signal from=east, inner sep=1pt, minimum height=6pt]
\tikzstyle{uy}=[draw, signal, fill=yellow, signal to=north, signal from=south, inner sep=1pt, minimum width=6pt]
\tikzstyle{dy}=[draw, signal, fill=yellow, signal to=south, signal from=north, inner sep=1pt, minimum width=6pt]

\tikzstyle{arrow}=[->]
\tikzstyle{very thick}=[-, line width=1pt, tikzit draw=red]
\tikzstyle{pointille}=[dashed, -]
\tikzstyle{red}=[-, draw=red]
\tikzstyle{blue}=[-, draw=blue]
\tikzstyle{green}=[-, draw=green]
\tikzstyle{strike}=[-, tikzit draw={rgb,255: red,191; green,0; blue,64}, strike through]
\tikzstyle{strike'}=[-, tikzit draw=cyan, strike bend]
\tikzstyle{dashed arrow}=[->, tikzit draw=green, draw=black, dashed]
\tikzstyle{reprise}=[-, line width=2pt, tikzit draw={rgb,255: red,255; green,191; blue,191}]

\newtheorem{lemma}{Lemma}

\title{A note on diagonal gates in SZX-calculus}

\author{Titouan Carette\\ Université de Lorraine, CNRS, Inria, LORIA, F 54000 Nancy, France\\ titouan.carette@loria.fr}

\begin{document}

\maketitle

This note describes how the the scalable ZXH calculus can be used to represent in a compact way the quantum gates that are diagonal in the computational basis. This includes controlled and multi-controlled Z gates, their generalizations, respectively graph and hypergraph operators, and also phase gadgets.

We don't present the different graphical languages in details. I invite the reader to report to \cite{coecke2011interacting} for ZX-calculus, \cite{backens2018zh} for ZH and \cite{carette2019szx} and \cite{carette2020colored} for the scalable notations.

However we use some unusual notations and conventions that are presented in the first section. In particular we use the well-tempered normalization of \cite{de2020well} and an unconvetional way to index the phases.

The exact interpretations of the generators are given in the first section, such that any suspicious reader can check by hand the soundness of the rewriting steps. We then introduce function arrows and use them to improve the representation of graph-states of \cite{carette2019szx}. We also gives a compact representation of hypergraph-states and phase gadgets. Finally we rephrase and extend a little bit the results of \cite{de2020fast} and use an algebraic technique to check the spider nest indentity of \cite{de2020fast}.

\tableofcontents

\section{Notations and semantics}

The ZX-calculus was first introduced for wires of size one in \cite{coecke2011interacting}. We use the extension to big wires called scalable ZX-calculus\cite{carette2019szx}. Our type system keeps track of the organisation of qubits into registers. We denote the type of a register of $n$ qubits $[n]$. Our types are defined by: $a,b\df [n], n\in \mathbb{N} | a \otimes b$. We write $[m]^n$ for the tensor product of $n$ registers of size $m$ with the convention $[m]^0=[0]$ the empty register. To each type we associate a Hilbert space: $\interp{[n]}\df(\mathbb{C}^2)^{\otimes n}$ and $\interp{a\otimes b}\df \interp{a}\otimes\interp{b}$ with the convention $\interp{[0]}=(\mathbb{C}^2)^{\otimes 0}=\mathbb{C}$. A map of quantum registers $f:a\to b$ is a linear map between the corresponding Hilbert spaces together with a reorganisation of the qubit registers. So in the interpretation we also keep track of the data organisation: $\interp{f:a\to b}\df \left( f: \interp{a}\to \interp{b} ,a,b \right)$. We use the contracted notation: $\interp{f:a\to b}= f : a \to b $ where $f$ is in fact typed $\interp{a}\to \interp{b}$. A register of type $[n]$ is represented by a wire of size $n$. Wires have interpretation:

\begin{center}
	\begin{tabular}{ll}
		
		$\interp{\begin{tikzpicture}
	\begin{pgfonlayer}{nodelayer}
		\node [style=none] (31) at (1, 0) {};
		\node [style=none] (32) at (0, 0) {};
	\end{pgfonlayer}
	\begin{pgfonlayer}{edgelayer}
		\draw [style=very thick] (31.center) to (32.center);
	\end{pgfonlayer}
\end{tikzpicture}
}\df \sum\limits_{x\in \textbf{2}^n}\ket{x}\bra{x}:[n]\to [n]$ & $\interp{\input{swap.tikz}}\df \sum\limits_{x\in \textbf{2}^n,y\in \textbf{2}^m}\ket{yx}\bra{xy}:[n]\otimes [m]\to [m]\otimes [n]$\\
		&\\
		$\interp{\input{cup.tikz}}\df \sum\limits_{x\in \textbf{2}^n}\ket{xx}:[0]\to [n]\otimes [n]$ & $\interp{\input{cap.tikz}}\df \sum\limits_{x\in \textbf{2}^n}\bra{xx}: [n]\otimes [n]\to [0]$
		
	\end{tabular}
\end{center}

The dividers and gatherers have interpretation:

\begin{center}
	\begin{tabular}{cc}
		
		$\interp{\input{ddd.tikz}}\df \sum\limits_{x\in \textbf{2}^{n+m}}\ket{x}\bra{x}:[n+m] \to [n]\otimes [m]$ & $\interp{\input{ggg.tikz}}\df \sum\limits_{\textbf{x}\in \textbf{2}^{n+m}}\ket{x}\bra{x}: [n]\otimes [m]\to[n+m]$
		
	\end{tabular}
\end{center}

We use generators from ZX and ZH calculus with the well tempered normalization of \cite{de2020well}. Since we are not concerned with completness here, we often write directly the scalars instead of their diagramatic representation. The green and red famillies of spiders and the yellow familly of harvestmen are indexed by a phase vector $a\in \mathbb{R}^k$ and depicted respectively:

\begin{center}
	\begin{tabular}{l}
		$\interp{\input{z.tikz}}\df 2^{k\frac{n+m-2}{4}}\sum\limits_{x\in \textbf{2}^k} e^{i\pi (x\cdot a)}\ket{x}^{\otimes n}\!\bra{x}^{\otimes m}: [k]^n\to [k]^m$\\
		\\
		$\interp{\input{x.tikz}}\df 2^{k\frac{2-n-m}{4}}\sum\limits_{x_i \in \textbf{2}^k} \prod\limits_{j=1}^{k}\frac{1+e^{i\pi \left(a_j + \sum\limits_{i=1}^{n+m}x_{i,j}\right)}}{2}\ket{x_1 \cdots x_n}\!\bra{x_{n+1} \cdots x_{n+m}}: [k]^n\to [k]^m$\\
		\\
		$\interp{\input{h.tikz}}\df 2^{-k\frac{n+m}{4}}\sum\limits_{x_i\in \textbf{2}^k} \prod\limits_{j=1}^{k} (1-2e^{i\pi a_j})^{\bigwedge\limits_{i=1}^{n+m} x_i}\ket{x_1 \cdots x_n}\!\bra{x_{n+1} \cdots x_{n+m}}: [k]^n\to [k]^m$
	\end{tabular}
\end{center}

Phase vectors have in fact values in $\left(\quot{\mathbb{R}}{2\mathbb{Z}}\right)^k$. Two phase vector equal modulo $2$ lead to the same semantics.

\section{Function arrows}

Given a boolean function $f:\textbf{2}^n \to \textbf{2}^m$ we naturally turn it into a linear map $\mathbb{C}^{\textbf{2}^n} \to \mathbb{C}^{\textbf{2}^m}$ and define a \textbf{function arrow}.

\begin{center}
	$\interp{\input{farrow.tikz}}=\ket{x}\mapsto 2^{\frac{m-n}{4}}\ket{f(x)}$
\end{center}

Any function arrow satisfies: $\input{fapply0.tikz}=\begin{tikzpicture}
	\begin{pgfonlayer}{nodelayer}
		\node [style=grdot] (66) at (0, 0) {$f(x)$};
		\node [style=none] (71) at (1.5, 0) {};
	\end{pgfonlayer}
	\begin{pgfonlayer}{edgelayer}
		\draw [style=very thick] (66) to (71.center);
	\end{pgfonlayer}
\end{tikzpicture}
$ ,$\input{fgerase0.tikz}=\begin{tikzpicture}
	\begin{pgfonlayer}{nodelayer}
		\node [style=ggdot] (69) at (1, 0) {};
		\node [style=none] (71) at (0, 0) {};
	\end{pgfonlayer}
	\begin{pgfonlayer}{edgelayer}
		\draw [style=very thick] (71.center) to (69);
	\end{pgfonlayer}
\end{tikzpicture}
$ and $\input{fgcopy0.tikz}=\input{fgcopy1.tikz}$.

Furthtermore: $\left(\input{farrow.tikz}\right)^\dagger=\input{dagfarrow.tikz}$.

\subsection{Matrix arrows}

Some famillies of function arrows enjoy additional properties that translate into graphical equations. In this subsection we detail two of them which can both be indexed by $\{0,1\}$-matrices.

\subsubsection{Red matrix arrow}

A function $f:2^n \to 2^m$ can be seen as a map $f:\mathbb{F}_2^n \to \mathbb{F}_2^m$, if this map is $\mathbb{F}_2$-linear then it can be described by a matrix $A\in \mathcal{M}_{m\times n}\left(\mathbb{F}_2\right)$. The \textbf{red matrix arrows} indexed by $A$ is then defined by $\input{rarrow.tikz}\df \input{farrow.tikz}$. Those arrows have been extensively studied in \cite{carette2019szx} and \cite{carette2020colored}. We recall here the main properties of red matrix arrows. Being $\mathbb{F}_2$-linear translates to:

\begin{center}
	$\input{Arcoerase0.tikz}=\begin{tikzpicture}
	\begin{pgfonlayer}{nodelayer}
		\node [style=grdot] (69) at (0, 0) {};
		\node [style=none] (71) at (1, 0) {};
	\end{pgfonlayer}
	\begin{pgfonlayer}{edgelayer}
		\draw [style=very thick] (71.center) to (69);
	\end{pgfonlayer}
\end{tikzpicture}
\quad$ and $\quad\input{Arcocopy0.tikz}=\input{Arcocopy1.tikz}$.
\end{center}

They interact with the dividers and gatherers as:

\begin{center}
	$\input{rowred0.tikz}=\input{rowred1.tikz}$ and $\input{colred0.tikz}=\input{colred1.tikz}$.
\end{center}

Most of the properies of red matrix arrows can be sumed up into one meta rule:

\begin{center}
$\input{span0.tikz}=\input{span1.tikz} \Leftrightarrow \quad Im\begin{pmatrix}
C\\D
\end{pmatrix}=Ker\begin{pmatrix}
A&B
\end{pmatrix}$
\end{center}

With $k\df dim\left(Ker\begin{pmatrix}
C\\D
\end{pmatrix}\right)$ and $h\df dim\left(coKer\begin{pmatrix}
A&B
\end{pmatrix}\right)$.

Red matrix arrows were the first motivation to the definition of scalable notations in \cite{chancellor2016graphical}, where they are applied to the design error correcting codes. The connection with codes is partially explained by the following result.

\begin{lemma}
	$C$-$Not$ circuits are equivalent to invertible red matrix arrows.
\end{lemma}

\begin{proof}
	A C-Not gate between two qubits of a register corresponds to a transvection matrix arrows. Those matrices span the special linear group which is, in caracteristic two, the same as the linear group of invertible matrix arrows.
\end{proof}

\subsubsection{Yellow matrix arrow}

As noted in \cite{carette2020colored}, the possibility to index arrows by matrices is linked to a bialgebra structure. If the red/green bialgebra leads to red matrix arrows, the yellow/green bialgebra gives us another familly of matrix arrows over the boolean semi-ring $\mathbb{B}\df \left(\{0,1\}, \land, \lor\right)$. A function $f:2^n \to 2^m$ can be seen as a map $f:\mathbb{B}^n \to \mathbb{B}^m$, if this map is a homorphism of $\mathbb{B}$-semi module, that is $f(a\land b)=f(a)\land f(b)$ and $f(1)=1$, then it can be described by a matrix $A\in \mathcal{M}_{m\times n}\left(\mathbb{B}\right)$. The \textbf{yellow matrix arrow} indexed by $A$ is then defined by $\input{yarrow.tikz}\df \input{farrow.tikz}$. Being a homorphism of $\mathbb{B}$-semi module translates to:

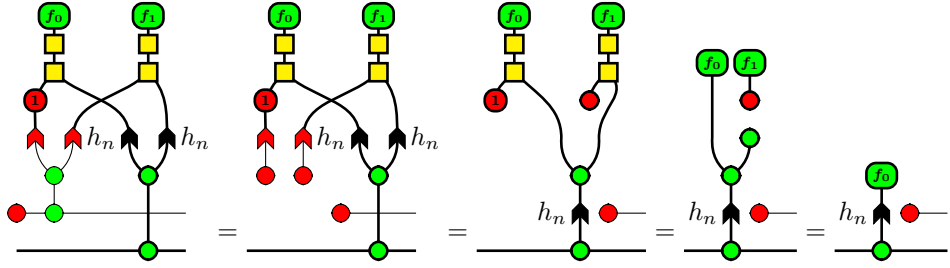
\begin{center}
	$\input{yArcoerase0.tikz}=\begin{tikzpicture}
	\begin{pgfonlayer}{nodelayer}
		\node [style=grdot] (69) at (0, 0) {$1$};
		\node [style=none] (71) at (1, 0) {};
	\end{pgfonlayer}
	\begin{pgfonlayer}{edgelayer}
		\draw [style=very thick] (71.center) to (69);
	\end{pgfonlayer}
\end{tikzpicture}
$ and $\input{yArcocopy0.tikz}=\input{yArcocopy1.tikz}$.
\end{center}

Yellow matrix have less properties than red ones. However we still have: 

\begin{center}
	$\input{rowyel0.tikz}=\input{rowyel1.tikz}$ and $\input{colyel0.tikz}=\input{colyel1.tikz}$.
\end{center}

\section{Diagonal gates}

A \textbf{diagonal gate} is a unitary map that is diagonal in the computational basis. In other words, it is a linear map $\mathbb{C}^{2^n} \to \mathbb{C}^{2^n}$ such that for each $x\in \textbf{2}^n$, $\ket{x}$ are eigenvectors with module one eigenvalue. The diagonal gates on $n$ qubits form an abelian group isomorphic to $\mathbb{U}^n$. A \textbf{phase function} is a semi-boolean function $f:2^n\to \mathbb{R}$. To each phase function we can associate a diagonal gate $e^{i\pi f}:[n]\to [n]$ defined by $e^{i\pi f}: \ket{x}\mapsto e^{i\pi f(x)}\ket{x}$. The correspondence is not one to one, $e^{i\pi f}=e^{i\pi g}$ if and only if there is a $k\in \mathbb{N}$ such that $f=g+2k$. We have $e^{i\pi (f+g)}=e^{i\pi f}e^{i\pi g}$. Graphically, being a diagonal gate is equivalent to being a phase of the scaled green spider. That is to be a unitary satisfying : $\input{phase0.tikz}=\input{phase1.tikz}$. We can represent graphically any diagonal gate with function arrows and phase functions. We define the \textbf{set function} $h_n:\textbf{2}^n\to \textbf{2}^{\textbf{2}^n}$ as: $\forall x,s\in \textbf{2}^n, h(x)_s\df \delta_{x=s}$. We can also define $h_n$ inductively by $h_0()=1$ and $h_{n+1}(x_0x')=\begin{cases}
h_n(x')0\cdots 0 \text{ if } x_0=0 \\
0\cdots 0h_n(x') \text{ if } x_0=1 
\end{cases}$.

Denoting $f$ the vector $(f(x))_{x\in2^n}$, $e^{i\pi f}$ is pictured: $\input{phase2.tikz}$. It has clearly the form of a phase. Moreover, any map of this form is unitary:

\begin{center}
	$\input{uphase0.tikz}=\input{uphase1.tikz}=\input{uphase2.tikz}=\input{uphase3.tikz}=\input{uphase4.tikz}=\input{uphase5.tikz}$.
\end{center}

We can check it is a correct representation of $e^{2\pi f}$:

\begin{center}
	$\input{dphase0.tikz}=\input{dphase1.tikz}=\input{dphase2.tikz}=\input{dphase3.tikz}=
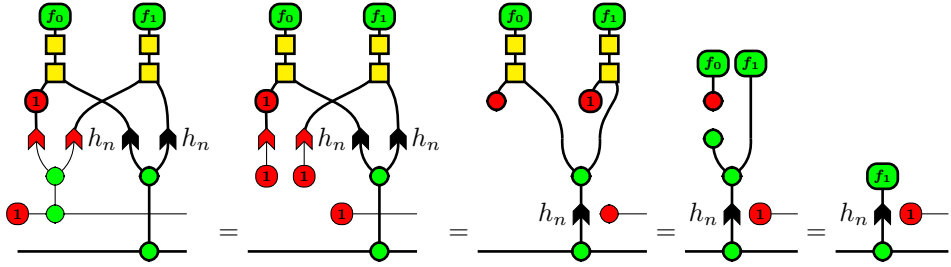
\begin{tikzpicture}
	\begin{pgfonlayer}{nodelayer}
		\node [style=none] (86) at (2, 0) {};
		\node [style=grdot] (92) at (1, 0) {$x$};
		\node [style=none] (104) at (0, 1) {$e^{i\pi f(x)}$};
	\end{pgfonlayer}
	\begin{pgfonlayer}{edgelayer}
		\draw [style=very thick] (92) to (86.center);
	\end{pgfonlayer}
\end{tikzpicture}
$
\end{center}

Defining $f_0(x)\df f(0x)$ and $f_1(x)\df f(1x)$ we have:
\begin{center}
	
	$\input{flem3a.tikz}=\input{flem2a.tikz}$
\end{center}

We can see that:

\begin{center}
	$\input{flem8.tikz}=\input{flem9.tikz}=\input{flem10.tikz}=\input{flem11.tikz}=\input{flem12.tikz}$
\end{center}

and

\begin{center}
	$\input{flem18.tikz}=\input{flem19.tikz}=\input{flem20.tikz}=\input{flem21.tikz}=\input{flem22.tikz}$
\end{center}

We now focus on famillies of diagonal gates that admit specific representations.

\section{Graph operator}

Graph-states were among the first example of application of the scalable notation in \cite{carette2019szx}. We provide here a much nicer representation. A C-Z gate has interpretation $\scalebox{0.5}{$\begin{pmatrix}
	1&0&0&0\\0&1&0&0\\0&0&1&0\\0&0&0&-1
	\end{pmatrix}$}$ and graphical representation $\input{czgt.tikz}$. A composition of C-Z gates on $n$ qubits is called a \textbf{graph-state operator}. In fact given a graph $\left(V,E\right)$ with $|V|=n$, the associated graph operator $G:\mathbb{C}^{2^n}\to \mathbb{C}^{2^n}$ is defined as the composition of the C-Z gates on the qubits $i$ and $j$ for each $(i,j)\in E$. Usually we define the graph-state $G \ket{+}^{\otimes n}$ instead of the graph operator. See \cite{hein2006entanglement} for more informations on graph-states.

\subsection{Graphical representation}

If $G$ is a bipartite graph $\left(V_0,V_1,E\right)$ with $|V_0|=a$ and $|V_1|=b$ then we can write the graph operator $\input{bgraph.tikz}$, where $\Gamma$ is the biadjacency matrix of $G$ defined by $\Gamma_{i,j} = \delta_{(i,j)\in E}$. Here the red matrix arrow is applying C-Nots that Hadamards gates turn into C-Z as expected. Given a non-bipartite graphs $\left(V,E\right)$, we build a bipartite graph $\left(V,V,E'\right)$. We fix an ordering of the vertices in $V$ and define: $E'\df \{ (i,j), i,j\in V, i>j\text{ and }(i,j)\in E\}$. The ordering ensures that each edge appears only once. Then the biadjacency matrix $\Gamma$ is upper triangular and satisfies $\Gamma+\Gamma^t= A$, the adjacency matrix of $G$. Thus we call $\Gamma$ the \textbf{half adjacency matrix} of $G$. We then fuse together the copies of the same vertex in the bipartite graph operator with green nodes: 
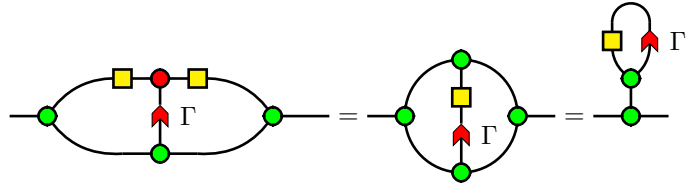
\vspace{-0.5cm}
\begin{center}
	 $\input{bgraph0.tikz}=\input{bgraph1.tikz}=\input{bgraph2.tikz}$.
\end{center}

We recognize the typical form of a diagonal operator.

We denote graph state operators $\input{gbox.tikz}\df \input{bgraph2.tikz}$. We have:

\begin{center}
	$\input{compgraph0.tikz}=\input{compgraph1.tikz}$ and $\input{tensgraph0.tikz}=\input{tensgraph1.tikz}$
\end{center}

We can now provide graphical versions of the properties of graph operators.

\subsection{Stabilizer}

The graph state operator corresponding to the graph $\left(V,E\right)$ satisfies $G\ket{0}=\ket{0}$:

\begin{center}
	$\input{ngraph0.tikz}=\input{ngraph1.tikz}=\input{ngraph2.tikz}=\input{ngraph3.tikz}=\begin{tikzpicture}
	\begin{pgfonlayer}{nodelayer}
		\node [style=none] (115) at (3, 0) {};
		\node [style=grdot] (116) at (2, 0) {};
	\end{pgfonlayer}
	\begin{pgfonlayer}{edgelayer}
		\draw [style=very thick] (116) to (115.center);
	\end{pgfonlayer}
\end{tikzpicture}
$
\end{center}

and for each vertex $i\in V$, $G\circ X_i =X_i \circ Z_{(\Gamma +\Gamma^t) i} \circ  G$:

\begin{center}
	$\input{egraph0.tikz}=\input{egraph1.tikz}=\input{egraph2.tikz}=\input{egraph3.tikz}=\input{egraph4.tikz}=\input{egraph5.tikz}$.
\end{center}

\subsection{Local complementation}

Graph states can be modified by applying phase gates locally on the vertices. More precisely we have $X_u(-\frac{\pi}{2}) \circ Z_{N_u}(\frac{\pi}{2})\circ G\ket{+}=(G*u)\ket{+}$ where $G*u$ is the graph $G$ locally complemented in $u$. Our proof is a scalable reformulation in scalable notations of the one of \cite{duncan2009graph}. We denote $T$ the half adjency matrix of the complete graph. We assume that the vertex in $G$ are ordered such that we first have $u$, then the neighborhood of $u$ and then the other vertices. Denoting $+$ the $\frac{1}{2}$ phase and $-$ the $-\frac{1}{2}$ phase we then want to prove:

\begin{center}
	$\input{lcgraph0.tikz}=\input{lcgraph1.tikz}$
\end{center}

First we consider the case of star graphs. In this situation the half adjacency matrix that we call $S_d$ has the form: $\begin{pmatrix}
\begin{matrix}
0\\ 1 \\ \vdots \\ 1
\end{matrix}& (0)
\end{pmatrix}$. We have:
\vspace{-1.5cm}
\begin{center}
	$\input{sgraph0.tikz}=\input{sgraph1.tikz}=\input{sgraph2.tikz}=\input{sgraph3.tikz}$.
\end{center}

Admiting the triangle lemma proved in \cite{duncan2009graph}: $\input{trig0.tikz}=\input{trig1.tikz}$. We prove the generalize verison: $\input{lem0.tikz}=\input{lem1.tikz}$.

\begin{proof}
	For $n=1$, we have $\input{lem0.tikz}=\begin{tikzpicture}
	\begin{pgfonlayer}{nodelayer}
		\node [style=none] (231) at (0.5, -1.5) {};
		\node [style=none] (232) at (0.5, 0) {};
		\node [style=none] (237) at (0.5, 1.5) {};
		\node [style=had] (243) at (0.5, 0) {};
	\end{pgfonlayer}
	\begin{pgfonlayer}{edgelayer}
		\draw (232.center) to (231.center);
		\draw (237.center) to (243);
	\end{pgfonlayer}
\end{tikzpicture}
=\input{lem1.tikz}$. For $n=2$, we have $\input{lem0.tikz}=\input{trig0.tikz}$ and $\input{lem1.tikz}=\input{trig1.tikz}$, this is exactly the triangle lemma. For $n\geq 3$:
	
	\begin{center}
		$\input{lem3.tikz}=\input{lem4.tikz}=\input{lem5.tikz}=\input{lem6.tikz}=\input{lem7.tikz}=\input{lem8.tikz}=\input{lem9.tikz}=\input{lem10.tikz}=\input{lem11.tikz}$.
	\end{center}
\end{proof}

Using the generalized triangle lemma we have:

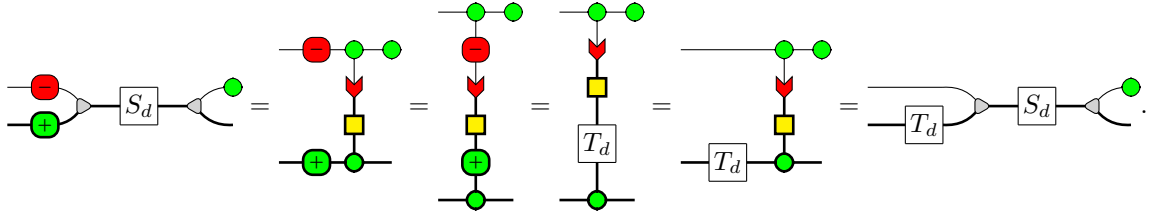
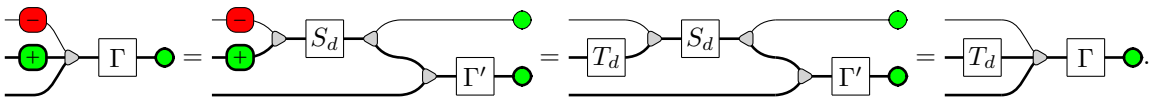
\begin{center}
	$\input{stgraph0.tikz}=\input{stgraph1.tikz}=\input{stgraph2.tikz}=\input{stgraph3.tikz}=\input{stgraph4.tikz}=\input{stgraph5.tikz}$.
\end{center}

And finally:

\begin{center}
	$\input{figraph0.tikz}=\input{figraph1.tikz}=\input{figraph2.tikz}=\input{figraph3.tikz}$.
\end{center}

\section{Hypergraph operator}

When an edge in a graph operator is represented by a controlled Z gate on two vertices, an hyperedge in a hypergraph operator is represented by a multi-controlled Z gates on a subset of the vertices. They can be easily represented in $ZH$-calculus by an Hadamard node. A phase function corresponding to an hyperedge is defined by $\forall s,x\in 2^n \xi_{s}(x)\df \delta_{s\leq x}$.\\

We can depict $\xi_s$ as: $\input{phase3.tikz}$ where $s$ is the characteristic vector of the subset $s$. A composition of such gates is called a \textbf{hypergraph operator}. We can represent them compactly in scalable notations. The matrix $H_n\in \{0,1\}^{2^n} \times \{0,1\}^n$ is defined inductively by: $H_1=\begin{pmatrix}
0\\1
\end{pmatrix}$ and $H_{n+1}\df \begin{pmatrix}
\begin{matrix}
0\\\vdots\\0
\end{matrix}& H_n\\ \begin{matrix}
1\\\vdots\\1
\end{matrix}& H_n
\end{pmatrix}$.

An hypergraph operator corresponding to the hypergraph $G$ on $n$ vertices is entirely defined by the phase function $g:2^n\to \mathbb{R}$ such that $g(s)=\delta_{s\in G}$. $H_n$ is a stack of all possible $s\in 2^n$. We can then draw the hypergraph operator $G$: $\input{phase4.tikz}$. We can generalize this to any non $\{0,1\}$ phase functions $f:\textbf{2}^n \to \mathbb{R}$ of the form: $f(s)=\sum\limits_{x\in \textbf{2}^n} a_x \xi_s (x)$. We have:

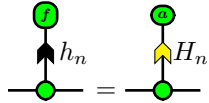
\begin{center}
	$\input{phase2.tikz}=\input{phase8.tikz}$
\end{center} 

We will see later that in fact any function arrow admits a representation of this form.

\section{Phase gadgets}

Phase gadgets are diagonal gates depicted by $\input{phase10.tikz}$. They can be represented by the phase functions defined by $\forall s,x\in 2^n, \Omega_{s}(x)\df x\cdot s$. $\Omega_s$ is depicted: $\input{phase5.tikz}$. Like for hypergraphs operator we can use the set matrix to depict a composition of phase gadget. We can represent any phase function $f:\textbf{2}^n \to \mathbb{R}$ of the form: $f(s)=\sum\limits_{x\in \textbf{2}^n} a_x \Omega_s (x)$. We have:

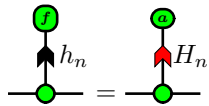
\begin{center}
	$\input{phase2.tikz}=\input{phase9.tikz}$
\end{center} 

We now apply those representations to graphical transforms.

\section{Graphical transforms}

The graphical Fourier theory was introduced in \cite{kuijpers2019graphical}. It was then stated in a mix of bang-boxes and ellipsis. We restate it here in an ellipsis free way using scalable notations. We hope this new presentation allows to grasp more clearly the underlying phenomena. The theory can be extended to other semi-boolean transforms. We do it there with the Möbius transform.

\subsection{Walsh Fourier transform}

Defining $\chi_s(x)\df 1-2\Omega_s(x)$, the $\chi_s$ form an orthonormal basis of $\mathbb{R}^{\textbf{2}^n}$ with respect to the scalar product $\braket{f}{g}\df \frac{1}{2^n}\sum\limits_{x\in\textbf{2}^n}f(x)g(x)$. The Walsh Fourier transform of a phase function $f$ is defined by $\hat{f}\df\frac{1}{2^n}\sum\limits_{s\in 2^n}f(s)$:

\begin{center}
	$f=\sum\limits_{s\in \textbf{2}^n} \hat{f}(s)\chi_s =\sum\limits_{s\in \textbf{2}^n} \hat{f}(s)\left(1-2\Omega_s\right)=\sum\limits_{s\in \textbf{2}^n} \hat{f}(s) -2\sum\limits_{s\in \textbf{2}^n} \hat{f}(s)\Omega_s=f(\emptyset)-2\sum\limits_{s\in \textbf{2}^n} \hat{f}(s)\Omega_s$
\end{center}

We can use this formula to rewrite the diagonal gate associated to $f$ as a composition of phase gadgets.

\begin{lemma}
	For all $f:\textbf{2}^n \to \mathbb{R}$
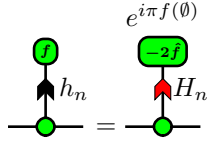
	\begin{center}
		$\input{phase2.tikz}=\input{phase6.tikz}$
	\end{center}
\end{lemma}

We provide a graphical proof.

\begin{proof}

Denoting $f_0(x)\df f(0x)$ and $f_1(x)\df f(1x)$. Let the Walsh matrix be $W\df \frac{1}{2}\begin{pmatrix}
1&1\\1&-1
\end{pmatrix}$. The Walsh Fourier transform of a phase function $f:2^n\to \mathbb{R}$ is $\hat{f}\df W^{\otimes n} f$.

\begin{center}
	$\begin{pmatrix}
	\hat{f}_0 \\
	\hat{f}_1
	\end{pmatrix}=W^{\otimes n+1} \begin{pmatrix}
	f_0 \\
	f_1
	\end{pmatrix}=\frac{1}{2}\begin{pmatrix}
	I_n& I_n\\ I_n & -I_n
	\end{pmatrix}\begin{pmatrix}
	W^{\otimes n}f_0 \\
	W^{\otimes n}f_1
	\end{pmatrix}=\frac{1}{2}\begin{pmatrix}
	I_n& I_n\\ I_n & -I_n
	\end{pmatrix}\begin{pmatrix}
	\widehat{f_0} \\
	\widehat{f_1}
	\end{pmatrix}=\begin{pmatrix}
	\frac{\widehat{f_0}+\widehat{f_1}}{2} \\
	\frac{\widehat{f_0}-\widehat{f_1}}{2}
	\end{pmatrix}$
\end{center}

So:

\begin{center}
	\begin{tabular}{ccccccc}
		$\hat{f}_0=\frac{\widehat{f_0}+\widehat{f_1}}{2}$&& $\hat{f}_1=\frac{\widehat{f_0}-\widehat{f_1}}{2}$
		&&
		$\widehat{f_0}=\hat{f}_0+\hat{f}_1$ && $\widehat{f_1}=\hat{f}_0 -\hat{f}_1$
	\end{tabular}
\end{center}

	By induction, for $n=1$ this is direct, for $n\geq 2$: We show $\input{flem1.tikz}=\input{flem0.tikz}$  and  $\input{flem3.tikz}=\input{flem2.tikz}$.
	
	We have: $\input{flem1a.tikz}=\input{flem0a.tikz}$ so:
	
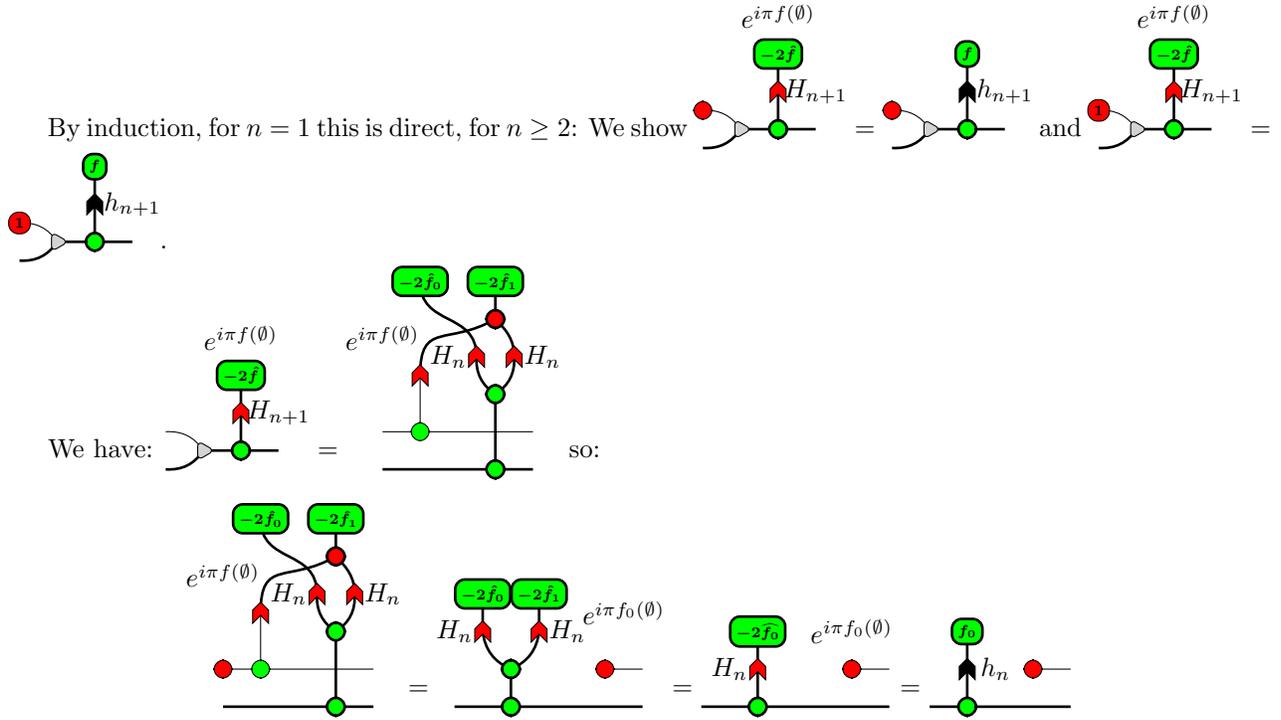
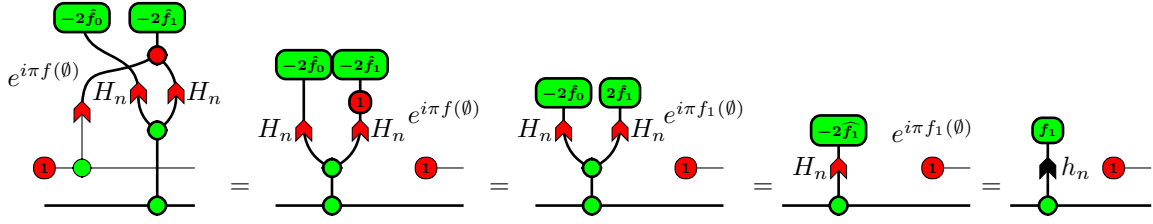
	\begin{center}
		$\input{flem4.tikz}=\input{flem5.tikz}=\input{flem6.tikz}=\input{flem7.tikz}$
	\end{center}

	and

	\begin{center}
		$\input{flem13.tikz}=\input{flem14.tikz}=\input{flem15.tikz}=\input{flem16.tikz}=\input{flem17.tikz}$
	\end{center}

\end{proof}

This equality has been proved in \cite{kuijpers2019graphical}.

\subsection{Möbius transform}

The $\xi_s$ also form a basis of $\mathbb{R}^{\textbf{2}^n}$ associated with the Möbius transform, see \cite{grabisch2016bases} for details. The Möbius transform of a phase function $f$ is defined by $\tilde{f}(x)=\sum\limits_{s\leq x}(-1)^{|x|+|s|}f(s)$. We have:

\begin{center}
	$f=\sum\limits_{s\in 2^n} \tilde{f}(s)\xi_s$
\end{center}

\begin{lemma}
	For all $f:\textbf{2}\to \mathbb{R}$:
	
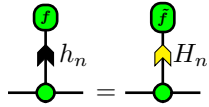
	\begin{center}
		$\input{phase2.tikz}=\input{phase7.tikz}$
	\end{center}
\end{lemma}

\begin{proof}
Let the Möbius matrix be $M\df \begin{pmatrix}
1&0\\-1&1
\end{pmatrix}$. The Möbius transform of a phase function $f:2^n\to \mathbb{R}$ is $\tilde{f}\df M^{\otimes n} f$.

\begin{center}
	$\begin{pmatrix}
	\tilde{f}_0 \\
	\tilde{f}_1
	\end{pmatrix}=M^{\otimes n+1} \begin{pmatrix}
	f_0 \\
	f_1
	\end{pmatrix}=\begin{pmatrix}
	I_n& 0\\ -I_n & I_n
	\end{pmatrix}\begin{pmatrix}
	M^{\otimes n}f_0 \\
	M^{\otimes n}f_1
	\end{pmatrix}=\begin{pmatrix}
	I_n& 0\\ -I_n & I_n
	\end{pmatrix}\begin{pmatrix}
	\widetilde{f_0} \\
	\widetilde{f_1}
	\end{pmatrix}=\begin{pmatrix}
	\widetilde{f_0} \\
	\widetilde{f_1}-\widetilde{f_0}
	\end{pmatrix}$
\end{center}

So:

\begin{center}
	\begin{tabular}{ccccccc}
		$\tilde{f}_0=\widetilde{f_0}$ && $\tilde{f}_1=\widetilde{f_1}-\widetilde{f_0}$&&$\widetilde{f_0}=\tilde{f}_0$ && $\widetilde{f_1}=\tilde{f}_0 + \tilde{f}_1$
	\end{tabular}
\end{center}
	By induction, for $n=1$ it is direct, for $n\geq 2$: We show $\input{mlem1.tikz}=\input{flem0.tikz}$  and  $\input{mlem3.tikz}=\input{flem2.tikz}$.
	
	We have: $\input{mlem1a.tikz}=\input{mlem0a.tikz}$ so:
	
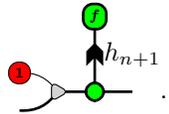
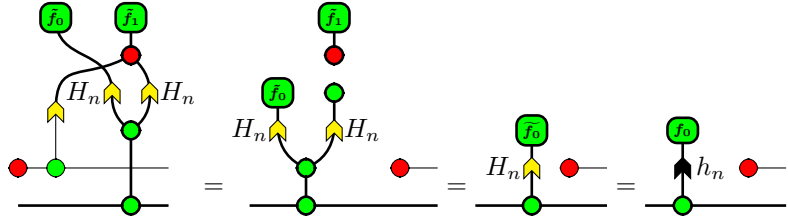
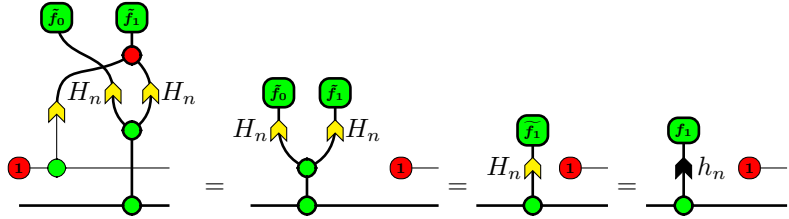
	\begin{center}
		$\input{mlem4.tikz}=\input{mlem5.tikz}=\input{mlem6.tikz}=\input{flem7.tikz}$
	\end{center}
	
	and
	
	\begin{center}
		$\input{mlem14.tikz}=\input{mlem15.tikz}=\input{mlem16.tikz}=\input{flem17.tikz}$
	\end{center}

\end{proof}

\subsection{Spider nests}

A spider nest identity is a composition of spiderlike diagrams, typically generalized hyperedges and phase gadgets, with one big spider and a lot of very small ones. Furthermore this composition must be the identity. We end this note by deriving some of them from graphical transforms. We restrict to \textbf{symmetric phase functions}, that is $f(x)$ only depends of the Hamming weight of $x$. We write $F:\textbf{n}\to \mathbb{R}$ the function such that $f(x)=F(|x|)$.

\subsection{Binomial transform}

The Möbius transform of a symmetric semi-boolean function is the \textbf{binomial transform}: 

\begin{center}
	$\tilde{f}(x)=\sum\limits_{s\leq x} f(s)(-1)^{|s|+ |x|}=(-1)^{|x|}\sum\limits_{k=0}^{|x|}\binom{|x|}{k}(-1)^k F(k)$
\end{center}

We define $\tilde{F}(m)\df\sum\limits_{k=0}^{m}\binom{m}{k} (-1)^{m-k} F(k)$. We have $\tilde{f}(x)=\tilde{F}(|x|)$ and $F(m)=\sum\limits_{k=0}^{m}\binom{m}{k}\tilde{F}(k)$.

We recover the spider nest identity of \cite{munsonand} by computing the Mobius transform of the phase function
$G(m)=\alpha\frac{1-(-1)^m}{2}$ which represent one phase gadget on $n$ qubits with phase $\alpha$. 

\begin{center}
	$\tilde{G}(m)=\sum\limits_{k=0}^{m}\binom{m}{k} (-1)^{m-k} G(k)=\frac{\alpha}{2}\left(\sum\limits_{k=0}^{m}\binom{m}{k} (-1)^{m-k}-\sum\limits_{k=0}^{m}\binom{m}{k} (-1)^{m-k}(-1)^{k}\right)=\frac{\alpha}{2}\left(\delta_{m=0}-(-2)^m\right)$.
\end{center}

$\tilde{G}(0)=0$ so we have no floating scalars. For $k\geq 1$, $\tilde{G}(k)=\alpha(-2)^{m-1}$. Setting $\alpha=\frac{1}{4}$ we see that only the first terms $\tilde{G}(1)=\frac{1}{4}$, $\tilde{G}(2)=\frac{-1}{2}$ and $\tilde{G}(3)=1$ are relevant for the phase gate.\\

This have been proved by induction in \cite{munsonand}.

\subsection{Kravchuk transform}

The case of the Fourier transform is more complex:

\begin{center}
	$\hat{f}(x)=\frac{1}{2^n}\sum\limits_{s\in \textbf{2}^n} f(s)(-1)^{s\cdot x}=\frac{1}{2^n}\sum\limits_{k\in \textbf{n}}F(k)\sum\limits_{|s|=k}(-1)^{s\cdot x}$.
\end{center}

$\sum\limits_{|s|=k}(-1)^{s\cdot x}$ is also a symmetric boolean function equal to the \textbf{Kravchuk polynomial}: 

\begin{center}
	$\mathcal{K}^n_k(|x|)\df\sum\limits_{j=0}^{k}\binom{|x|}{j}\binom{n-|x|}{k-j}(-1)^j$.
\end{center}
 To see this, consider $j$ as the number of ones in comon between $x$ and $s$. The Kravchuk polynomials satisfy: $\binom{n}{m}\mathcal{K}_k^n(m)=\binom{n}{k}\mathcal{K}_m^n(k)$ and $\sum\limits_{i=0}^{n}\binom{n}{i}\mathcal{K}_k^n(i)\mathcal{K}_l^n(i)=2^n \binom{n}{k}\delta_{k=l}$.

The Walsh Fourier transform of a symmetric semi-boolean function is then the \textbf{Kravchuk transform}: $\hat{F}(m)\df \frac{1}{2^n}\sum\limits_{k=0}^{n}F(k)\mathcal{K}^n_k(m)$. We have: $\hat{f}(x)=\hat{F}(|x|)$ and $F(m)=\sum\limits_{k=0}^{n}\hat{F}(k)\mathcal{K}^n_k(m)$. See \cite{canteaut2005symmetric} for details on transforms of symmetric semi-boolean functions.

We can compute the transform of the phase function $H(m)\df \beta\delta_{m=n}$ which corresponds to the generalised hyperedge on $n$ qubits with phase $\beta$.

\begin{center}
	$\hat{H}(m)=\frac{1}{2^n}\sum\limits_{k=0}^{n}H(k)\mathcal{K}_k^n(m)=\frac{1}{2^n}\sum\limits_{k=0}^{n}\beta\delta_{k=n}\mathcal{K}_k^n(m)=\frac{\beta}{2^n}\mathcal{K}_n^n(m)=\frac{\beta}{2^n}(-1)^m$.
\end{center}

Combining this result with the spider nest identity of the previous section it is possible to derive the spider nest identity from \cite{de2020fast} as it is done in \cite{munsonand}. We end this note by giving an alternative proof by inversion. We first sketch a method to check spider nests identity. We want to show that for a symmetric phase function $\hat{s}:\textbf{2}^n \to \mathbb{R}$:

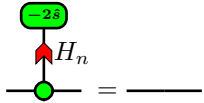
\begin{center}
	$\input{invk0.tikz}=\begin{tikzpicture}
	\begin{pgfonlayer}{nodelayer}
		\node [style=none] (86) at (2, 0) {};
		\node [style=none] (92) at (0, 0) {};
		\node [style=none] (96) at (1, 0) {};
	\end{pgfonlayer}
	\begin{pgfonlayer}{edgelayer}
		\draw [style=very thick] (92.center) to (96.center);
		\draw [style=very thick] (96.center) to (86.center);
	\end{pgfonlayer}
\end{tikzpicture}
$
\end{center}

We know $\hat{S}$ by reading the coefficients in the phase gadgets. We compute $S$ using the inversion formula. Then we check if all values of $S$ are equal modulo $2$. If it is the case this means that the corresponding phase gadget is $e^{i\pi S(0)}I_n$. But this scalar is exactly the one appearing in the graphical Fourier transform. So simplifying on both sides gives us exactly what we want.

We apply this method to the spider nest identity of \cite{de2020fast}. Here, only the Kravchuk polynomials for $k=0,1,2,3$ and $n$ are needed:

\begin{center}
	\begin{tabular}{cccc}
		$\mathcal{K}^n_0(m)=1$&
		$\mathcal{K}^n_1(m)=-2m+n$&
		$\mathcal{K}^n_2(m)=2m^2 -2nm+\frac{n^2 -n}{2}$&$\mathcal{K}^n_n(m)=(-1)^m$\\
		&&&\\
		\multicolumn{4}{c}{$\mathcal{K}^n_3(m)=-\frac{4}{3}m^3 +2nm^2  +(-n^2 +n -\frac{2}{3})m+\frac{n^3 - 3n^2 + 2n}{6}$}
	\end{tabular}
\end{center}

We want to inverse the phase function:

\begin{center}
	$\hat{S}(0)=0$ $\hat{S}(1)=\frac{(n-2)(n-3)}{16}$ $\hat{S}(2)=-\frac{n-3}{8}$ $\hat{S}(3)=\frac{1}{8}$ $\hat{S}(n)=-\frac{1}{8}$  $\hat{S}(k)=0$ for $ k\neq 0,1,2,3,n$.
\end{center}

\begin{lemma}
	Forall $m\in \interp{0,n}$, $S(m)= S(0)\mod 2$
\end{lemma}

\begin{proof}

\begin{center}
	\begin{align*}
		S(m)&=\sum\limits_{k=0}^{n}\hat{S}(k)\mathcal{K}_k^n(m)=\hat{S}(1)\mathcal{K}_1^n(m)+\hat{S}(2)\mathcal{K}_2^n(m)+\hat{S}(3)\mathcal{K}_3^n(m)+\hat{S}(n)\mathcal{K}_n^n(m)\\
		&=\frac{(n-2)(n-3)}{16}\mathcal{K}_1^n(m)-\frac{n-3}{8}\mathcal{K}_2^n(m)+\frac{1}{8}\mathcal{K}_3^n(m)-\frac{1}{8}\mathcal{K}_n^n(m)\\
		&=\frac{-m^3}{6} +\frac{3m^2}{4} -\frac{5m}{6}+ \frac{n^3}{48}-\frac{n^2}{8}+\frac{11n}{48} -\frac{(-1)^m}{8}
	\end{align*}
	
\end{center}

Our goal is to check that $S(m) \mod 2$ doesn't depends on $m$. Thus, we only need the part of $S(m)$ that depends in $m$.

\begin{center}
	$S'(m)\df\frac{-m^3}{6} +\frac{3m^2}{4} -\frac{5m}{6}-\frac{(-1)^m}{8}$.
\end{center} 

Since $S'(0)=-\frac{1}{8} \mod 2$, we want to check that for each $m\in \mathbb{N}$, $S'(m)\equiv -\frac{1}{8} \mod 2$. To do so we write $m=12k+l$ with $k\in \mathbb{N}$ and $l\in \interp{0,11}$.
We obtain:

\begin{center}
	$S'(12k+l)=-288k^3 - 72lk^2 + 108k^2 - 6kl^2 + 18kl -10k- \frac{l^3}{6} +\frac{3l^2}{4}-\frac{5l}{6}-\frac{(-1)^{l}}{8}$
\end{center}

We see that $S'(12k+l) = - \frac{l^3}{6} +\frac{3l^2}{4}-\frac{5l}{6}-\frac{(-1)^{l}}{8} \mod 2$, this only depends on $l$. Thus we can just check that for each $l\in \interp{0,11}$, $- \frac{l^3}{6} +\frac{3l^2}{4}-\frac{5l}{6}-\frac{(-1)^{l}}{8}= -\frac{1}{8} \mod 2$  (which is true).

\end{proof}

\bibliography{scal.bib}

\end{document}